\documentclass[pre,twocolumn,superscriptaddress,aps,notitlepage,10pt]{revtex4-1}
\PassOptionsToPackage{usenames,dvipsnames}{color}
\usepackage{paralist}
\usepackage{amsmath,amsfonts,amssymb,amsthm,graphicx,bbm,enumerate,times,color}
\usepackage{mathtools}
\usepackage[normalem]{ulem}

\newtheorem{theorem}{Theorem}

\newtheorem{lemma}[theorem]{Lemma}
\newtheorem{corollary}[theorem]{Corollary}

\definecolor{christian}{rgb}{0,0,0}
\newcommand{\cg}[1]{{\color{christian} #1}}

\definecolor{rodrigo}{rgb}{0,0,0}
\newcommand{\ro}[1]{{\color{rodrigo}#1}}
\definecolor{henrik}{rgb}{0,0,0}

\newcommand{\je}[1]{{\color{black}#1}}

\newcommand{\he}[1]{{\color{henrik}#1}}
\newcommand{\mc}[1]{\mathcal{#1}}

\newcommand{\Var}{\mathrm{Var}}

\newcommand{\e}{\mathrm{e}}
\newcommand{\ii}{\mathrm{i}}

\newcommand{\tr}{\mathrm{Tr}} 
\newcommand{\Tr}{\mathrm{Tr}} 

\newcommand{\one}{\mathbbm{1}}

\newcommand{\R}{\mc{R}}
\newcommand{\deff}{d^{\mathrm{eff}}}
\newcommand{\psecond}{p^{2^{\mathrm{nd}}}}
\newcommand{\T}{\mc{T}}

\newcommand{\norm}[1]{\left\Vert #1 \right\Vert}

\newcommand{\ket}[1]{\left.|{#1}\right.\rangle}

\newcommand{\bra}[1]{\left.\langle{#1}\right.|}

\newcommand{\ketbra}[2]{\ket{#1} \!\! \bra{#2}}

 \newcommand{\proj}[1]{\ketbra{#1}{#1}}

\renewcommand{\hat}{}

\usepackage{hyperref}
\begin{document}
\title{What it takes to shun equilibration}

\author{R.~Gallego}
\author{H.~Wilming}
\author{J.~Eisert}
\affiliation{Dahlem Center for Complex Quantum Systems, Freie Universit{\"a}t Berlin, 14195 Berlin, Germany}
\author{C.~Gogolin}
\affiliation{ICFO-Institut de Ciencies Fotoniques, The Barcelona Institute of Science and Technology, 08860 Castelldefels (Barcelona), Spain}
\begin{abstract}
Numerous works have shown that under mild assumptions unitary dynamics
 inevitably leads to equilibration of physical expectation values if many
 energy eigenstates contribute to the initial state.  Here, we consider systems
 driven by arbitrary time-dependent Hamiltonians as a protocol to prepare
 systems that do \emph{not} equilibrate. We introduce a measure of the
 resilience \ro{against equilibration of such states} and show, under \cg{natural} assumptions, that in order to increase the
 resilience against equilibration of a given system, one needs to possess a
 resource system which itself has a large resilience. \ro{In this  way,} we
 establish a new link between the theory of equilibration and resource theories
 by quantifying the resilience against equilibration and the resources that are
 needed to produce it.  We connect these findings with insights into local
 quantum quenches and investigate the (im-)possibility of formulating a second 
 law of equilibration, by studying how resilience can be either only
 redistributed among subsystems, if these remain completely uncorrelated, or in turn
 created in a catalytic process if subsystems are allowed to build up some correlations.
\end{abstract}
\maketitle
\section{Introduction}
When a complex quantum system is prepared in a pure state, it evolves in time
under its Hamiltonian indefinitely and, after sufficiently long passage of
time, even returns to its initial state if it is finite dimensional.
Nevertheless, many physical observables quickly approach a stationary value and
only rarely significantly deviate from this again.
Over recent years, great progress has been made in understanding this
phenomenon. On the one hand,  rigorous proofs for such behavior under very general
assumptions have been found \cite{Reimann08,Linden2010,Reimann12,Gogolin2016,Short2012}. 
On the other hand, it is the basis of a plethora of studies
of quantum many-body systems out of equilibrium \cite{1408.5148,PolkovnikovReview}, 
motivated from a condensed matter perspective.

A central quantity for proofs of equilibration in pure state quantum
statistical mechanics is the so-called \emph{effective dimension}.
Roughly speaking, it measures with how many energy eigenstates a given initial state
has a significant overlap. Provided that a system has a large effective
dimension, one can prove under very mild additional assumptions---such as a
non-degenerate energy gaps condition---that expectation values of 
subsystems \cite{Linden09} and observables
with a reasonably small number of outcomes equilibrate
\cite{Reimann08,Reimann2012,Reimann12}
(see Ref.\ \cite{Gogolin2016} for a review). In addition, systems with a large effective
dimension can be proven to equilibrate fast for certain observables with few
outcomes \cite{Malabarba2014} and systems fluctuate around equilibrium slowly
\cite{Linden2010}, which forces weakly coupled systems to decohere
\cite{Gogolin2010}. We note, however, that it is an important open problem to
prove general results that show how long it takes for a system to equilibrate, 
despite considerable recent progress on this question
\cite{Short2012,Goldstein2013,Goldstein2015,Garcia-Pintos2015a,Reimann2016,Masanes2013,Wilming2017,DeOliveira2017}.

It has been argued that the effective dimension is large in realistic
situations \cite{Reimann08,Reimann2012,Reimann12,Linden09}.
The core of such
arguments is that, due to the exponential growth of the dimension of the
Hilbert space of many-body systems, it is unreasonable to expect that an
initial state will have overlap with only few of them (except in gapped systems
at very low temperature).
Indeed, if the system has a reasonably large
energy uncertainty and its energy distribution is centered around some
value (no strongly bimodal distribution), then one can argue that
the effective dimension will grow exponentially with the system size,
hence ensuring equilibration for large enough systems.
While such arguments \ro{suggest} quite convincingly that it 
may be difficult
to bring a large many-body system to a state that would
\emph{not} equilibrate (and hence 
is not expected to happen spontaneously), they
do not quantify this difficulty.

Here, we provide such a quantification by analyzing the question from a
fresh perspective: this is the perspective of the emergent field of
quantum thermodynamics (see Refs.~\cite{Goold2015b,Nicole} for 
reviews), specifically that of resource theories, which is increasingly becoming important in that context.
In this mindset, we consider the task of preparing states that can avoid equilibration
when provided with another quantum system as a resource and allowing
for a very general set of operations which includes arbitrary evolutions of both systems under arbitrary time-dependent Hamiltonians.
This analysis, employing tools from resource theories
of purity or coherence
\cite{Horodecki2003,Aberg2006,Gour2015,Baumgratz2014,Winter2016,Streltsov2016,Marvian2016,StructureCoherence},
leads to simple formulas expressing lower bounds on the amount of resources
that are needed in order to prepare systems with a \cg{low} effective dimension.
Indeed, these bounds imply---given reasonable assumptions on the form of the
resources---that it is difficult to bring an arbitrarily large interacting
many-body system to a state which does not equilibrate, in the precise sense
that it would require a diverging amount of resources. This provides stronger
quantitative grounds to the heuristic arguments laid out above, suggesting that
systems should realistically be expected to have a large effective dimension
and thus equilibrate. We also connect these insights to the description of
quantum quenches.

We then go further and study in a general fashion how the property of resisting
equilibration behaves under composition of systems. That is, we treat such
property as a quantifiable resource and investigate how it can be created or
re-shuffled among subsystems when one operates globally on them by making them
interact and evolve under arbitrary time-dependent Hamiltonians.  In
particular, we show that it is possible to find certain systems, referred to as
catalysts, which can be used to bring other systems to a state resisting
equilibrate, whereas the catalyst remains itself invariant in the process.  The
implications of this insight are discussed in relation with what could be
called a second law of equilibration, putting forward also further challenges
within the study of equilibration properties of many body systems.

This work is structured as follows: We first briefly review equilibration
results in pure state quantum statistical mechanics. In
Section~\ref{sec:stationary}, we present our framework to prepare states
that do not equilibrate and provide general bounds for the resources
required in the case of stationary initial states. We then study in
Section~\ref{sec:correlations} how resilience can be either only
redistributed among subsystems, if these remain completely
uncorrelated, or in turn created in a catalytic process if subsystems
are allowed to build up arbitrarily weak correlations. Finally, in
Section~\ref{sec:non-stationary}, we extend our results to the case of
non-stationary states under additional assumptions on the process that
brings the system out of equilibrium.

\section{Equilibration in pure state quantum statistical mechanics}
Let us summarize some of the important equilibration results
\cite{Reimann08,Linden09,Reimann12,Gogolin2016} to highlight the relevance of
the effective dimension and set the notation. Consider a finite dimensional
quantum system evolving unitarily according to the Schr{\"o}dinger
equation through the states $\rho(t)$ starting from an initial state $\rho$ under a Hamiltonian $H$ with eigenvectors $\{\ket{E_k}\}_k$
and corresponding energies $\{E_k\}_k$. We assume that $H$ does not have any
degenerate energy gaps, i.e., $H$ is non-degenerate and any energy difference between
energy levels is unique. This condition is generically fulfilled for fully interacting Hamiltonians.
For simplicity we will henceforth assume that all
appearing Hamiltonians have this property. 

We denote by $\overline{\,\cdot\,}
\coloneqq \lim_{T\to\infty} 1/T \int_0^T \,\cdot\ \mathrm{d}t$ the infinite time-average.
Let $A$ be an arbitrary observable. Then its expectation value equilibrates to
the expectation value in the dephased or time-averaged state $\omega_H(\rho)
\coloneqq \sum_k p_k
\ketbra{E_k}{E_k}  = \overline{{\rho(t)}}$ with $p_k
\coloneqq \bra{E_k}\rho\ket{E_k}$, whenever the \emph{effective dimension}
\begin{align}\label{eq:def:effectivedimension}
  \deff_H(\rho) \coloneqq \frac{1}{\sum_k p_k^2} = \frac{1}{\tr(\omega_H(\rho)^2)}
\end{align}
is sufficiently large. This is captured by the bound
\begin{align} \label{eq:equilibrationonaverageforexpectationvalues}
	{\Var_H}(A,\rho) \coloneqq \overline{\Tr\big(A\,(\rho(t) - \omega_H(\rho)) \big)^2} \leq \frac{\|A\|^2}{\deff_H(\rho)}.
\end{align}
A similar
bound also holds in terms of the second largest
eigenvalue of $\omega_H(\rho)$ instead of the effective dimension
\cite{Reimann08,Reimann12}, but the formulation in terms of the effective
dimension is more suitable for our purposes as it will allow us to state our
results in terms of entropies 
\footnote{See Supplemental Material \ref*{effective_dimensions} at [URL
will be inserted by publisher] for equilibration results in terms of the second largest eigenvalue of $\omega_H(\rho)$ instead of the effective dimension and a comparison to the effective dimension and von~Neumann entropy of $\omega_H(\rho)$.}.

If the combined system is (mentally or physically) split into the tensor product of two
subsystems $S$ and $B$, it can be shown that the whole state $\rho^S(t) \coloneqq \Tr_B(\rho(t))$ of the subsystem
$S$ is, most of the time, close in trace distance to the local equilibrium state $\omega_H^S(\rho) \coloneqq \tr_B(\omega_H(\rho))$ in the sense that \cite{Linden09,Short2012}
\begin{equation}\label{eq:equilibration}
	\overline{\mc{D}\big(\rho^S(t),\omega_H^S(\rho)\big)} \leq 
	\frac{1}{2}\sqrt{\frac{d_S^2}{\deff_H(\rho)}},
\end{equation}
where $d_S$ denotes the dimension of
the Hilbert space of $S$. The trace distance measures the (single shot)
distinguishability under arbitrary observables, so, if the right hand side of
the above inequality is small, this means that for most $t$ there exists not a
single observable that would allow one to tell apart $\rho^S(t)$ from
$\omega_H^S(\rho)$. The above inequalities show that systems whose initial
states have a large effective dimension appear to equilibrate to great
precision and it is hence natural to ask: How difficult is it to prepare a
state of a quantum many-body system that can avoid equilibration?

\section{The cost of avoiding equilibration}
\label{sec:stationary}
Let us first, based on the role played by the effective dimension in the
equilibration bounds \eqref{eq:equilibrationonaverageforexpectationvalues} and
\eqref{eq:equilibration}, define the \emph{resilience (against equilibration)}
of a $d$ dimensional system initially in state $\rho$ evolving 
under the Hamiltonian $H$ as
\begin{equation}\label{eq:def_resilience}
 \R(\rho,H) \coloneqq \log \left(\frac{d}{\deff_H(\rho)} \right).
\end{equation}
Having a high resilience is a necessary condition for avoiding equilibration in
the long run.  To further illustrate its relation to equilibration, consider
for example $\rho(t)$ being the maximally mixed state for all $t$, which is, and
stays, equilibrated under any Hamiltonian. In this case $\R(\rho, H)=0$.  The
other extreme is given by a $\rho$ with $1< \deff_H(\rho) \ll d$\ro{, which} in general does not equilibrate. In this case one obtains
$\R(\rho,H) \approx \log d \propto n$, where $n$ is the size (number of
subsystems) of the whole system.

Note that a large resilience is a necessary, but not a sufficient, condition
for a system to avoid equilibration, as can be seen by taking $\rho(t)$ to be an
eigenstate of Hamiltonian. For our work this is not problematic, as we will be
concerned with obtaining lower bounds on the resources needed to prepare states
with a large resilience. Our results show that states that do not equilibrate
are difficult to prepare, whereas situations such as the case of $\rho(t)$ being
an eigenstate only show that in addition some states that do equilibrate can
also be difficult to prepare.

We now present a very general scenario that models possible preparations of
quantum states out of equilibrium and which includes the common settings of
quenches, ramps, and other control protocols.  We are given two quantum systems
$Q$ and $R$ in the product state $\hat\sigma^{Q} \otimes \hat\sigma^{R}$ and
with non-interacting Hamiltonian $H_i^{QR} \coloneqq H_i^Q \otimes \one^R +
\one^Q \otimes H_i^{R}$ (here we call a Hamiltonian non-interacting if and only
if it is exactly of this form).  For now,  let us assume that
both systems are stationary, i.e., 
$[\hat\sigma^{Q} \otimes \hat\sigma^{R},H_i^{QR}] = 0$.
Our aim is now to bring the system $Q$ out of equilibrium.
For this we are allowed \ro{to drive the evolution of the systems by changing the global Hamiltonian} in a completely arbitrary fashion \ro{and making $Q$ and $R$ interact}.
\ro{More formally, we let} the global system \ro{evolve} unitarily under
a time-dependent Hamiltonian $H^{QR}(t)$ for $t \in [t_i,t_f]$ from the initial Hamiltonian $H^{QR}(t_i)=H_i^{QR}$ to the final Hamiltonian $H^{QR}(t_f)\coloneqq H_f^{QR}= H_f^{Q} \otimes \one^{R} + \one^{Q}\otimes H_f^{R}$, which must again be non-interacting.
\he{Given a pair of initial and final Hamiltonians, one can in principle always find a trajectory $\alpha$ which implements on the state any possible unitary $U_\alpha$. 
Even going beyond this set of operations, for sake of generality, we allow the
use of a source of randomness to implement \ro{a} trajectory $\alpha$ between the two fixed Hamiltonians $H_i^{QR}$ and $H_f^{QR}$ with some probability $p_\alpha$.
On average the state on $QR$ is then transformed as $\hat\sigma^{QR}\mapsto \sum_\alpha p_\alpha U_\alpha \hat\sigma^{QR} U_\alpha^\dagger$.   
All possible protocols bringing a system out of equilibrium in this way result in a unital map $\Lambda$ on the given quantum state, i.e., $\Lambda$ fulfills $\Lambda(\one)=\one$. This is in fact the only property we need for our result below. 
To summarize, the process of bringing the system out of equilibrium can be described abstractly as
\begin{align}\
	(\hat\sigma^{QR}, H^{QR}_i)\mapsto (\Lambda(\hat\sigma^{QR}),H^{QR}_f) \eqqcolon (\rho^{QR},H^{QR}_f),
\end{align}
with $\Lambda(\one)=\one$.}
Note that we allow for an unrealistically
high degree of control over the precise trajectories of Hamiltonians, which only makes our bounds stronger, since we are interested
in lower bounding the resources needed to prepare systems that do not
equilibrate. Hence, incorporating fewer restrictions on the set of operations
can only strengthen our lower bounds.

Given the initial state and Hamiltonian and using the set of operations laid out above, the task is to optimally exploit
the resource system $R$ by choosing an optimal final Hamiltonian \he{and
time-dependent trajectories, resulting in the channel $\Lambda$,} to prepare a
state on $Q$ so that it does not equilibrate for the chosen final Hamiltonian.
Let us denote with $\rho_{t}^Q$ the time evolved state of $Q$ under $H_f^Q$ after $\Lambda$ has been applied.
Its initial condition is then
\begin{equation}\label{eq:transition}
 \rho^Q = \Tr_R\big(\Lambda(\hat\sigma^Q \otimes \hat\sigma^R)\big).
\end{equation}
In order to prepare $Q$ so that it does not equilibrate one needs to obtain a
sufficiently large resilience $\R(\rho^Q,H_f^Q)$. We are now in a position to
formulate our first main result, which provides general upper-bounds on
that resilience in terms of the initial state and Hamiltonian.

\begin{theorem}[Resources for preparing resilient states from stationary ones]\label{thm:mainresult}
Consider a system $Q$ and resource $R$ given in the product state $\hat
\sigma^{Q} \otimes \hat\sigma^{R}$ stationary with respect to the
non-interacting Hamiltonian $H_i^{QR}$. Then, for any final non-interacting
Hamiltonian $H_f^{QR}$ and operation of the form \eqref{eq:transition}, the
resulting state $\rho^Q$ fulfills
\begin{equation}\label{eq:mainresult}
 \R(\rho^Q,H_f^Q) - \R(\hat\sigma^Q,H_i^Q) \eqqcolon \Delta\R^Q \leq \R(\hat\sigma^R,H_i^R).
\end{equation}
\end{theorem}

Before providing the proof of Theorem~\ref{thm:mainresult} let us discuss its
interpretation and consequences. First note that the resulting resilience
$\R(\rho^Q,H_f^Q)$ is upper-bounded by a function of the initial state and
Hamiltonian only. If those initial conditions provide a small enough bound,
then it follows that no protocol from the very general set of operations
considered above can bring $Q$ from an equilibrating state to a state
which does not equilibrate. Secondly, note that $\Delta \R^Q$ is the change of
the resilience of $Q$.  Hence, Eq.~\eqref{eq:mainresult} states that the change
in the resilience of $Q$ is upper-bounded by the resilience present in the
resource $R$.  This implies that the resilience against equilibration behaves
like a ``thermodynamic resource'': in order to increase the resilience of
$Q$ (so that it has a chance of avoiding equilibration) we need to possess
another system $R$ from which to take this resilience.
While this result is phrased in a resource-theoretical
language, it readily applies to paradigmatic situations from the context of 
quantum systems out of equilibrium such as quenches, ramps, and general control protocols, underlining its physical significance.
\he{
\begin{corollary}[Local quenches]\label{thm:localquench} Consider a
many-body system $Q$ in a initial stationary state $\hat \sigma^Q$ and initial
Hamiltonian $H_i^Q$. Let $H_f^Q$ be any Hamiltonian on $Q$ and $\Phi$ be any quantum channel (not necessarily unital) that acts non-trivially
only on a subsystem $X$ of dimension $d_X$. Then
\begin{equation}\label{eq:localquench}
	\R(\Phi(\sigma^Q),H_f^Q) - \R(\hat\sigma^Q,H_i^Q)  \leq \log d_X.
\end{equation}
\end{corollary}

This statement implies that local quenches, applied to some small subsystem, such as flipping a spin,
will alter
the resilience also only very little \footnote{See Supplemental Material \ref*{localquench} at [URL
will be inserted by publisher]  for the proof of the corollary using a Stinespring dilation.}.
A similar bound can be expected to hold if the Hamiltonian on $Q$ is
only locally modified in a time-dependent manner for a finite time and
with a bound on its norm, due to Lieb-Robinson bounds \cite{Lieb1972}.}

Let us consider now a further illustrating example \ro{of the general bound provided by Theorem \ref{thm:mainresult}}. Let the state $\hat\sigma^Q$ be a micro-canonical state for the initial Hamiltonian $H_i^Q$, with an energy window containing
$K(n)$ eigenstates. We allow the number of eigenstates to depend on the size
$n$ of the system arbitrarily, but typically one expects $K(n)$ to grow approximately exponentially
with $n$. Suppose that the resource consists of \ro{a} $m$-partite system with
local Hilbert space dimension $D$. In this case
one obtains
\begin{align}
	\deff_{H_f^Q}(\rho^Q)\geq \frac{K(n)}{D^m},
\end{align}
which implies that the size of the resource $m$ has to diverge if one wants to
prepare a state that can avoid equilibration, as long as the number of states
within the micro-canonical window grows with $n$. 

The proof of Theorem~\ref{thm:mainresult} is based on properties of
\emph{R\'{e}nyi-divergences} \cite{Tomamichel2016} (a similar proof can also be
carried out using Tsallis-entropies \cite{Tsallis1988}). For any two quantum
states $\hat\sigma$ and $\hat\tau$ the R\'{e}nyi-divergences
are defined to be
\begin{align}\label{eq:defrenyi}
 D_\alpha(\hat\sigma \| \hat\tau) \coloneqq \frac{1}{1-\alpha}\log\left(\hat\sigma^{\alpha}\,\hat\tau^{1-\alpha}\right),
\end{align}
with the corresponding R\'enyi-$\alpha$ entropies defined as $S_\alpha(\rho)\coloneqq \log(d) - D_\alpha(\rho\| \one/d)$.
The R\'{e}nyi-divergences for $\alpha\in[0,2]$ satisfy the data-processing inequality, i.e., for any
quantum channel $\Phi$ it holds that
\begin{equation}\label{eq:dpi}
 D_{\alpha}(\Phi(\hat\sigma)\| \Phi(\hat\tau)) \leq D_\alpha(\hat\sigma\|\hat\tau),
\end{equation}
and  they are additive on product states, in that
\begin{align}
 D_\alpha(\hat\sigma_A\otimes\hat\sigma_B \| \hat\tau_A\otimes\hat\tau_B) =
 D_\alpha(\hat\sigma_A\| \hat\tau_A) + D_\alpha(\hat\sigma_B \| \hat\tau_B).
\end{align}
Combining Eqs.~\eqref{eq:def:effectivedimension}, \eqref{eq:def_resilience}, and
\eqref{eq:defrenyi} one can see that one can write the resilience in terms of
the R\'{e}nyi-divergences of the dephased state as
\begin{align}\label{eq:connectiontsalliswithdeff}
	\R(\rho,H)= D_2\big(\omega_H(\rho)\| \one/d\big).
\end{align}
This quantity and variants thereof have been studied in the context of the
resource theory of coherence and purity \cite{Gour2015,Streltsov2016}.  Indeed,
$\R(\rho,H)$ can be understood as the purity of the time-averaged state of
$\rho$.
Using additivity and the data-processing inequality one can show the following
two properties:
First, for any state $\hat\sigma$ that is stationary with respect to \he{a Hamiltonian} $H_i$, any unital
channel $\Lambda$ and Hamiltonian $H_f$ one \he{has that
\begin{equation}
	\R(\rho,H_f)= \R(\Lambda(\hat\sigma),H_f) \leq \R(\hat\sigma,H_i).
\label{eq:monotonicity}
\end{equation}}
\he{Second, consider product initial states $\hat\sigma^Q \otimes \hat\sigma^R$ and a non-interacting Hamiltonian $H^{QR}_i$.
If $\hat\sigma^Q$ and $\hat\sigma^R$ are stationary under the Hamiltonians $H_i^Q$ and $H_i^R$, respectively, we have that the resilience is additive
\begin{equation}
 \label{eq:additivity_resilience}
 \R(\hat\sigma^Q \otimes \hat\sigma^R, H^{QR}_i) = \R(\hat\sigma^Q,H^Q_i)+\R(\hat\sigma^R,H^R_i).
\end{equation}}
Given these two properties, which are proven in the Supplemental Material
\footnote{See Supplemental Material \ref*{monotonicity_and_additivity} at [URL
will be inserted by publisher] for the proofs of Eqs.~\eqref{eq:monotonicity}
and \eqref{eq:additivity_resilience}.}, the proof of
Theorem~\ref{thm:mainresult} is straightforward.

\begin{proof}[Proof of Theorem~\ref{thm:mainresult}]
The concatenated map $\Tr_R \circ \Lambda$ is unital.
Hence, using \eqref{eq:transition} and \eqref{eq:monotonicity} one obtains that
\begin{equation}\label{eq:proof1}
\R(\rho^Q,H_f^Q) \leq \R(\hat\sigma^Q \otimes \hat\sigma^R,H_i^{QR}).
\end{equation}
In addition, $\hat\sigma^Q$ and $\hat\sigma^R$ are stationary states by assumption, hence one can use additivity \eqref{eq:additivity_resilience} which yields
\begin{equation}\label{eq:terms}
 \R(\rho^Q,H_f^Q) 
 \leq \R(\hat\sigma^Q,H_i^Q)+ \R(\hat\sigma^R,H_i^R) .
\end{equation}
\end{proof}

\section{The second law of equilibration and correlations}
\label{sec:correlations}
As we have discussed after Theorem~\ref{thm:mainresult}, the resilience against
equilibration behaves like a thermodynamic resource in the sense that it cannot
be created if it is not already present in another system.
Nonetheless, \eqref{eq:mainresult} does not specify how much of the resource's
resilience is spent in the process, it only quantifies how much must be present.
We now turn to this question and
show that it depends on the correlations created between $R$ and $Q$ in the process.

\he{We consider again stationary given states $\hat\sigma^Q\otimes\hat\sigma^R$ as well as initial and final
non-interacting Hamiltonians $H_i^{QR}$ and $H_f^{QR}$, respectively. We then}
implement an operation of the form \eqref{eq:transition}, with the
difference that we do not trace out the state of the resource, i.e., let
\begin{equation}
 \rho^{QR} \coloneqq \Lambda\big(\hat\sigma^Q \otimes \hat\sigma^R\big).
\end{equation}
From \he{the proof of} Theorem~\ref{thm:mainresult}, one obtains the bound
\begin{align}\label{eq:bound_correlated}
\R(\rho^{QR},H^{QR}_{\ro{f}}) \leq \R(\hat\sigma^Q,H^Q_{\ro{i}}) + \R(\hat\sigma^R,H^R_{\ro{i}}).
\end{align}
If the \he{resulting} state is product, $\rho^{QR} = \rho^{Q}
\otimes \rho^{R}$, one can show \footnote{See Lemma \ref{lemma:superadditivity} in the Supplemental Material at [URL
will be inserted by publisher].} that
\begin{align}
\R(\rho^Q,H^Q_{\ro{f}})+\R(\rho^R,H^R_{\ro{f}})\leq \R(\rho^Q \otimes\rho^R,H^{QR}_{\ro{f}}),
\end{align}
which implies in turn that $\Delta\R^{R}$, the change on the resilience of the resource $R$, satisfies
\begin{equation}\label{eq:second_law}
	\he{-\Delta \R^{R} = \R(\hat\sigma^R,H^R_{\ro{i}}) - \R(\rho^R,H^R_{\ro{f}}) \geq \Delta \R^Q .}
\end{equation}
This inequality can be interpreted as a kind of second law of equilibration
in the sense that to bring a system $Q$ in an initial state that can resist
equilibration, another system $R$ must have been brought (closer) to an equilibrating initial state.
In other words, the property of permanently remaining out of equilibrium, as measured by the
resilience, can only be re-shuffled between systems but not created.
Note however that \eqref{eq:second_law} relies on the assumption that the
produced state $\rho^{QR}$ is such that there are no correlations
between $Q$ and $R$. 
This assumption may in many physically reasonable situations be unjustified\cg{.
This leads us to the question of whether it is possible to derive a second
law of equilibration that holds in full generality or, whether in contrast it is
possible to make use of correlations between $Q$ and $R$ to violate \eqref{eq:second_law}}.
In the following theorem we show
that the latter is the case.
Indeed, this is shown in \cg{the} strongest possible form: $Q$ \cg{can be} brought to a
trajectory which does not equilibrate while $R$ remains unchanged \cg{and only
a vanishing amount of correlations between $Q$ and $R$ need to be created for this. 
Furthermore, this is possible for arbitrary large systems and even when the initial and final Hamiltonians are identical.}
\begin{theorem}[\he{No second law of equilibration}]\label{thm:no_second_law}
 Consider a family of $n$-partite many-body systems with increasing $n$.
 There are stationary states $\hat\sigma_{(n)}^Q$ \ro{and non-interacting  initial and final Hamiltonians $H^{QR}_{i,(n)}=H^{QR}_{f,(n)}\coloneqq H^{QR}_{(n)}$} such that for every $\epsilon>0$ there exist resource states $\hat\sigma_{(n)}^R$ and channels of the form $\Lambda^{(n)}(\cdot)=\sum_i p_i^{(n)} U_i^{(n)} \cdot U_i^{(n)}$ \ro{with the following properties:}
 \begin{enumerate}[(1)]\setlength{\itemsep}{0pt}\setlength{\parskip}{0pt}
 \item \label{enum:non_equilibrating} For all $n$, the resulting state
 \begin{align}
	 \rho^Q_{(n)}=\tr_R\left(\Lambda^{(n)}(\hat\sigma_{(n)}^Q \otimes \hat\sigma_{(n)}^R)\right)
 \end{align} on $Q$ does not equilibrate \ro{under $H^{Q}_{(n)}$}.
 \item \label{enum:resource_unchaged} For all $n$, the state of the resource remains unchanged upon application of the channel: 
  $\rho_{(n)}^R \coloneqq \Tr_Q\left(\Lambda^{(n)}(\hat\sigma_{(n)}^Q \otimes \hat\sigma_{(n)}^R)\right)=\hat\sigma_{(n)}^R$.
 \item \label{enum:final_correlations} For all $n$, the correlations between $R$
	and $Q$ as measured by the mutual information remain arbitrarily
	\he{small: 
	 \begin{align}
		 I(Q:R) \coloneqq D_1(\rho_{(n)}^{QR}\|\rho_{(n)}^Q \otimes
	 \rho_{(n)}^R) \leq \epsilon.
	 \end{align}}
 \item \label{enum:change_of_resiliance} The change in resilience of the systems $\Delta \R_{(n)}^Q$ diverges with $n\rightarrow \infty$.
 \end{enumerate}
\end{theorem}
The proof of Theorem~\ref{thm:no_second_law} relies on a very recent result of Ref.~\cite{Mueller2017} which establishes conditions on possible transitions within the resource theory of thermodynamics and we provide all details in the Supplemental Material \footnote{See the Supplemental
Material \ref*{role_of_correlations} at [URL
will be inserted by publisher] for the proof of Theorem~\ref{thm:no_second_law} and a proof that the resilience of $\sigma_{(n)}^R$ diverges as $\epsilon\rightarrow 0$.}. 

Let us now discuss the implications of points~\ref{enum:non_equilibrating}--\ref{enum:final_correlations} of Theorem~\ref{thm:no_second_law}. We stress that
these implications are independent of the particular measure for resilience
against equilibration. As mentioned above, the theorem implies the existence of
initial states, in particular given by micro-canonical
states $\hat\sigma_{(n)}^Q$, which are perfectly equilibrated and can be brought
to quantum states $\rho_{(n)}^Q$ that do not equilibrate, as expressed by property \ref{enum:non_equilibrating}
while not spending any resources. In the proof of Theorem~\ref{thm:no_second_law}, we construct explicitly observables $A^{(n)}$  
so that $\Var_{H^Q_{(n)}}(A^{(n)},\rho_{(n)}^Q) \geq c\,\|A^{(n)}\|$ where $c>0$ is a constant independent of $n$. 
This, by property \ref{enum:resource_unchaged}, can be done while the state of $R$ remains unchanged.

Even more surprisingly, as $R$ does not change
in the process, it is hence possible to iterate this procedure 
with a sequence of many sub-systems $Q_1,\ldots,Q_m$, \he{which do not interact with each other,} by
re-using each time the same resource $R$, which is thus
independent of $m$. This procedure brings \emph{each} of the arbitrarily many
sub-systems to a state that does not equilibrate. More
sharply put, the theorem has the surprising implication that given the
right \cg{resource} $\hat\sigma^R$ and an environment in the micro-canonical state---composed of an arbitrarily large number of $m$ non-interacting
subsystems---one can, without \emph{spending} any \cg{of that resource}, turn the entire environment 
into a state that is (and remains) nowhere in equilibrium.
Note that
this is in sharp contrast with the situation covered by
Theorem~\ref{thm:mainresult2} concerning \emph{fully-interacting} systems.
This highlights and clarifies the importance of interactions, even arbitrarily small ones, for equilibration.
The extreme contrast between the non-interacting and the arbitrarily
weakly-interacting case is ultimately a result of the fact that an arbitrarily
small weak interaction can lead to equilibration after a sufficiently long
time. It is important to stress that by \emph{non-interacting} systems we refer to completely uncoupled systems and not quasi-free fermionic or bosonic systems. 

\cg{Let us now relate Theorem \ref{thm:no_second_law} to the second law like inequality \eqref{eq:second_law} to better understand the influence of correlations}.
First note that in full generality, even if $R$ and $Q$ are correlated, one
can use Eqs. \eqref{eq:monotonicity} and
\eqref{eq:additivity_resilience} to derive the condition \eqref{eq:bound_correlated}, which can be straightforwardly rewritten as
\begin{align}
	-\Delta \R^R\geq \Delta \R^Q + \mathcal{C}(\rho^{QR},H^{QR}),
\end{align} 
where the term $\mathcal{C}(\rho^{QR},H^{QR}):= \R(\rho^{QR},H^{QR})-\R(\rho^{Q},H^{Q})-\R(\rho^{R},H^{R})$ 
provides the correction to \eqref{eq:second_law} due to correlations. Then, a
consequence of property \ref{enum:change_of_resiliance} is that
$\mathcal{C}(\rho^{QR},H^{QR})$ can be made negative and arbitrarily large in
absolute value, even if the correlations are arbitrarily small as measured in
mutual information, as shown by property \ref{enum:final_correlations}. The
negativity of $\mathcal{C}(\rho^{QR},H^{QR})$ captures that $\R$ does not
fulfill super-additivity, which is ultimately a consequence of the R\'{e}nyi-2
entropy not fulfilling sub-additivity
\cite{Nielsen2000}. 

Interestingly, these implications of Theorem~\ref{thm:no_second_law} for the lack of 
super-additivity of $\R$ can be extended to alternative definitions of the
resilience, some of which are based on entropies other than the R\'{e}nyi-2
entropy. In the Supplemental Material \footnote{See the Supplemental Material
\ref*{effective_dimensions} at [URL will be inserted by publisher] for
a discussion on possible alternative definitions of resilience and
possible improvements of equilibration bounds and \ref*{shannon} for
the relation between Theorem~\ref{thm:no_second_law} and
subadditivity} we investigate in detail these
alternative definitions. Further, we show that it is impossible to construct bounds of the form
of Eqs.~\eqref{eq:equilibrationonaverageforexpectationvalues} and
\eqref{eq:equilibration} in terms of the R\'enyi entropies with $\alpha\leq 1$,
which also excludes the usual von~Neumann entropy corresponding to $\alpha=1$.
Since the known bounds \eqref{eq:equilibrationonaverageforexpectationvalues}
and \eqref{eq:equilibration} are formulated in terms of $\alpha=2$ (recall that $\deff_H(\rho) = \e^{S_2(\omega_H(\rho))}$), 
this only leaves open the possibility to derive improved versions of
\eqref{eq:equilibrationonaverageforexpectationvalues} and
\eqref{eq:equilibration} in terms of R\'enyi entropies with $1<\alpha<2$, non of which could lead to a sub-additive notion of resilience.

\section{Non-stationary states}
\label{sec:non-stationary}
In previous sections we considered the case where the provided states $\hat\sigma^Q$ and $\hat\sigma^R$ are stationary, which is arguably the most natural one in the context considered.
One is given systems that are always perfectly equilibrated and shall use them to prepare a system that can avoid equilibration.
Nonetheless, it is also possible to incorporate non-stationary states into our formalism
by including an extra assumption on the set of operations.
In order to motivate this assumption, let us first consider a simple example without resource 
system $R$.
Consider only the system $Q$ provided in a \he{pure state $\hat \sigma^Q$ and Hamiltonian $H_i^Q$ so that
$\hat\sigma^Q=\ketbra{\psi^Q}{\psi^Q}$ with
\begin{equation}
\ket{\psi^Q}=\frac{1}{\sqrt{d_Q}}\sum_{k}\ket{E_k},
\end{equation}
where $d_Q$ is the Hilbert space dimension of $Q$.}
This state is highly non-stationary, but it has a null resilience against equilibration.
It is further easy to construct a unitary evolution $U$ so that $\rho^Q=\Lambda
(\hat\sigma^Q)=U \hat \sigma^Q U^\dagger$ has a resilience proportional to $d_Q$.
This can be done by simply applying a unitary rotation that leaves the state
in a superposition of two eigenstates of the chosen final Hamiltonian $H_f^Q$.
This shows that preparing states that
do not equilibrate from states such as
$\hat\sigma^Q$, which are not stationary, is perfectly possible within our set of operations, even without employing any
resource state $R$.
Note, however, that since $\hat\sigma^Q$ 
is not stationary, applying the same
unitary $U$ at a slightly different time $T$ will in
general lead to a very different state $\Lambda \circ \T_{T}^{H^Q_i}
(\hat\sigma^Q) \neq \rho^Q$, where we introduced the time-evolution operator
\begin{align}
	\T_t^H(\cdot) \coloneqq \e^{-\ii H t} \cdot \e^{\ii H t}.
\end{align}
Hence, the subsequent time evolution cannot repeatedly be prepared
unless one has unrealistically fine control over the time $T$ which determines when we initiate the time-dependent evolution on the system.
This is crucial if one wants to reproduce many instances of the experiment in order to gather statistics and verify the non-equilibrating dynamics.

Therefore, it is interesting to consider \cg{unital preparation procedures}
$\Lambda$ that do not require such fine control over $T$.
More precisely, and now considering the problem in 
more generality by including a system $R$,
the extra assumption on the set of operations is that for every $t$ there exist
a $t'$ so that
\begin{equation}\label{eq:t_independence}
 \Lambda \circ  \T_{t}^{H_i^{QR}} (\hat\sigma^Q \otimes \hat\sigma^R) = \T_{t'}^{H^{QR}_f} \circ \Lambda  
 (\hat\sigma^Q \otimes \hat\sigma^R)
\end{equation}
(note that this contains formally as a particular case the situation of
stationary states $\hat\sigma^Q \otimes \hat\sigma^R$, for which \eqref{eq:t_independence} is
fulfilled taking $t'=0$). \he{Further note that we can also allow here for a family of unital maps $\{\Lambda_{s}\}$,
possibly depending on some time $s>0$, each of which fulfills~\eqref{eq:t_independence}.}

To analyse the condition \eqref{eq:t_independence} let us first consider the
simpler case of $H_i^{QR}=H_f^{QR}\coloneqq H$. There exist at least two well
known classes of operations that fulfill the condition in this case.
The first class is given by the so-called covariant maps
defined as the maps that fulfill
$\Lambda_s \circ \T_t^H = \T^{H}_t \circ \Lambda_s$ for all $t,s>0$.
Such covariant maps can
be motivated from different perspectives. Argued from the perspective of
resource theories,
they correspond to the
set of operations that can be performed without a reference frame for time (see
Ref.~\cite{Bartlett2007} and references therein) and relate to physically
relevant scenarios such as Hamiltonian evolution under the rotating-wave
approximation \cite{Lostaglio2017}. \he{We also note that a similar condition
appears in the context of fluctuation relations in the resource theory of
thermodynamics \cite{Alhambra2016}.}
An important example from the theory of many-body non-equilibrium dynamics is constituted by Markovian
dynamics reflected by a dynamical semi-group $s\mapsto \Lambda_s(\rho) = e^{{\cal L}s}(\rho)$ with the property that 
${\cal L}([H,\cdot])= [H,{\cal L}(\cdot)]$, \he{corresponding to some
dissipative dynamics}. \he{In the case $H^{QR}_i\neq
H^{QR}_f$, the
condition generalizes to ${\cal L}([H^{QR}_i,\cdot])= [H^{QR}_f, {\cal
L}(\cdot)]$.} 
It is important to stress, however, that \eqref{eq:t_independence} is substantially 
weaker than full covariance as it must only be
fulfilled for the particular state $\hat\sigma^Q \otimes \hat\sigma^R$ and allows for $t\neq t'$.

The second class is given by operations fulfilling $\Lambda =
\Lambda'\circ \omega_{H_i^{QR}}$, which in turn implies that $\Lambda \circ \T^{H_i^{QR}}_t=\Lambda$.
Such operations first put the system into a fully time-averaged state and then act on that.

For non-stationary provided states $\hat\sigma^{QR}$, in addition to assumption \eqref{eq:t_independence}, we need to put an assumption on the Hamiltonians.
Namely, we impose the condition $\omega_{H_i^{QR}} = \omega_{H_i^Q} \otimes \omega_{H_i^R}$, which is
generically fulfilled and simply implies that letting the two non-interacting
and initially uncorrelated systems $Q$ and $R$ equilibrate does not generate
correlations between equilibrating observables.
For example, it follows if there are no distinct
eigenvalues $E^Q_k\neq E^Q_l$ of $H_i^Q$ and $E^R_m\neq E^R_n$ of $H_i^R$ so that
$E^Q_k-E^Q_l=E^R_m-E^R_n$ \footnote{We assumed for the results on equilibration to
hold that $H^Q$ and $H^R$ have no degeneracies.
Hence, this condition implies that $H^{QR}$ has also no degeneracies.
But Eq.~\eqref{eq:additivity_resilience} actually still holds even if both
$H^Q$ and $H^R$ have degeneracies, as long as adding them does not
create additional ones.}.
We are now ready to formulate generalization of Theorem~\ref{thm:mainresult} to
non-stationary $\hat\sigma^Q,\hat\sigma^R$ for unital maps
$\Lambda$ that fulfill condition \eqref{eq:t_independence}.

\begin{theorem}[Resources needed to prepare resilient states]\label{thm:mainresult2} 
	\he{Consider a system $Q$ and resource $R$ provided in the product
	state $\hat \sigma^{Q} \otimes \hat\sigma^{R}$ and with Hamiltonian
$H_i^{QR}$ fulfilling $\omega_{H_i^{QR}} = \omega_{H^Q_i}\otimes\omega_{H^R_i}$. 
Then, for any final non-interacting Hamiltonian
$H_f^{QR}$ and operation of the form \eqref{eq:transition} with $\Lambda$
fulfilling Eq.~\eqref{eq:t_independence}, the resulting state $\rho^Q$
fulfills}
\begin{equation}\label{eq:mainresult2}
 \R(\rho^Q,H_f^Q) - \R(\hat\sigma^Q,H_i^Q) \eqqcolon \Delta\R^Q \leq \R(\hat\sigma^R,H^R_i).
\end{equation}
\end{theorem}

In the proof one uses Eq.~\eqref{eq:t_independence} to show that for such channels the monotonicity relation
Eq.~\eqref{eq:monotonicity} is also true for non-stationary states. Using the
property $\omega_{H^{QR}_i} = \omega_{H^Q_i}\otimes\omega_{H^R_i}$ one shows that the
additivity relation Eq.~\eqref{eq:additivity_resilience} is also fulfilled for
non-stationary states $\hat\sigma^{QR}$ (see the Supplemental Material \cite{Note3} for details).
The rest of the proof is then completely analogous to that of Theorem~\ref{thm:mainresult}.
\\
\section{Conclusions and outlook}\label{sec:conclusions}
Using an approach inspired by so-called resource theories
\cite{Horodecki2002,Brandao2008,Brandao2013,Brandao2015,Rio2015,Coecke2016} we
have studied the cost of preparing states that do not equilibrate. We have
shown that in order to prepare a large interacting system in a state that can
withstand equilibration in the long run it is necessary to have
access to another large system that can withstand equilibration.
We have also
discussed readings of what could be called a second law  of
equilibration for dynamics under unitary evolutions, with an emphasis on the role of correlations. Our work
connects the recently emerging field of resource theories of coherence
\cite{Bartlett2007,Vaccaro2008,Gour2008,Marvian2014,Baumgratz2014,Lostaglio2015,Lostaglio2015a,Streltsov2016,Winter2016,StructureCoherence}
with the pure state quantum statistical mechanics approach, establishing a
bridge between two sub-fields of modern quantum thermodynamics.
Our results show that such a resource theoretical
approach can be used to derive new and general results about long-standing
problems \ro{in the} foundations of statistical mechanics via
information theoretic methods.

\paragraph*{Acknowledgements.}
The group at FUB thanks
the DFG (EI 519/7-1, GA 2184/2-1, CRC 183-project B02), the European Research Council (TAQ), and the EC (AQuS) for support. H.\ W. further thanks the Studienstiftung des Deutschen Volkes for support. 
C.\ G.~is supported by a Marie Skłodowska-Curie Individual Fellowships (IF-EF) program under GA: 700140 by the European Union and acknowledges financial support from the European Research Council
(CoG QITBOX and AdG OSYRIS), the Axa Chair in Quantum Information Science,
Spanish MINECO (FOQUS FIS2013-46768, QIBEQI FIS2016-80773-P and Severo Ochoa Grant No.~SEV-2015-0522), and the Fundaci\'{o} Privada Cellex, and Generalitat de Catalunya (Grant No.~SGR 874 and 875, and CERCA Program).

\bibliographystyle{apsrev4-1}
\bibliography{low_effective_dimension}

\cleardoublepage
\phantom{x}
\makeatletter
\newcommand{\manuallabel}[2]{\def\@currentlabel{#2}\label{#1}}
\makeatother
\manuallabel{monotonicity_and_additivity}{Section~A}
\manuallabel{localquench}{Section~B}
\manuallabel{role_of_correlations}{Section~C}
\manuallabel{effective_dimensions}{Section~D}
\manuallabel{shannon}{Section~E}

\appendix
\makeatletter
\onecolumngrid
\begingroup
\frontmatter@title@above
\frontmatter@title@format
\@title{} Supplemental Material

\endgroup
\vspace{1cm}
\twocolumngrid
\makeatother
\setcounter{page}{1}

\subsection{\ref*{monotonicity_and_additivity}: Monotonicity and additivity of the resilience}
Let us first show the monotonicity of the resilience (Eq.~\eqref{eq:monotonicity} in the main text).
For the proof of Theorem~\ref{thm:mainresult} one requires monotonicity of the resilience for unital channels $\Lambda$ and stationary given states $\hat\sigma^{Q}$ and $\hat\sigma^{R}$, as expressed before Eq.~\eqref{eq:monotonicity} in the main text.
For Theorem~\ref{thm:mainresult2} one requires monotonicity for non-stationary given states $\hat\sigma^{QR}$ but it is sufficient to prove it for unital channels $\Lambda$ which fulfill in addition
\begin{equation}\label{eq:supp:t_independence}
 \Lambda \circ  \T_{t}^{H_i^{QR}} (\hat\sigma^Q \otimes \hat\sigma^R) = \T_{t'}^{H^{QR}_f} \circ \Lambda  
 (\hat\sigma^Q \otimes \hat\sigma^R)
\end{equation}
Note that this condition is automatically fulfilled for any $\Lambda$ whenever $\hat\sigma^{QR}$ is stationary.
Hence, it suffices to show monotonicity of the resilience for maps fulfilling \eqref{eq:supp:t_independence}.
\he{In the following, we drop the superscripts $QR$ when not explicitly needed to simplify the notation.}

Because of the definition of $\omega_H$ (above
\eqref{eq:def:effectivedimension} in the main text), 
for any two fixed Hamiltonians $H_i$ and $H_f$, we can employ condition \eqref{eq:supp:t_independence} to show in a compact way that
\begin{align}
 (\omega_{H_f} \circ \Lambda \circ \omega_{H_i}) (\hat\sigma) &= 
 \left(\omega_{H_f} \circ \Lambda \circ \lim_{T \rightarrow \infty} \frac{1}{T}\int_{0}^{T} \mathrm{d}t \: \T_t^{H_i}\right) (\hat\sigma)\\
&= \lim_{T \rightarrow \infty} \frac{1}{T}\int_{0}^{T} \mathrm{d}t \: (\omega_{H_f} \circ \Lambda \circ \T_t^{H_i})(\hat\sigma)\\
&= \lim_{T \rightarrow \infty} \frac{1}{T}\int_{0}^{T} \mathrm{d}t \: (\omega_{H_f} \circ \T_{t'}^{H_f} \circ \Lambda) (\hat\sigma) \label{eq:tprime} \\
&= \lim_{T \rightarrow \infty} \frac{1}{T}\int_{0}^{T} \mathrm{d}t \: (\omega_{H_f} \circ \Lambda) (\hat\sigma) \label{eq:tprimegone} \\
 \label{eq:monotonicity_proof0}&= (\omega_{H_f} \circ \Lambda)(\hat\sigma),
\end{align}
where \eqref{eq:tprime} follows from condition \eqref{eq:supp:t_independence},
\eqref{eq:tprimegone} follows from the fact that $\omega_H \circ \T_t^{H}=\omega_H
$ for all $t>0$, and \eqref{eq:monotonicity_proof0} simply because the integrand
does not depend on $t$.
With this we can now show monotonicity as
\begin{align}
 R(\hat\sigma,H_i) &= D_2 (\omega_{H_i} (\hat\sigma) \| \one/d)\\
 \label{eq:monotonicity_proof1}&= D_2 (\omega_{H_i} (\hat\sigma) \| \omega_{H_i} (\one/d)) \\
\label{eq:monotonicity_proof2}&\geq D_2 ((\omega_{H_f} \circ \Lambda \circ \omega_{H_i}) (\hat\sigma) \| (\omega_{H_f} \circ \Lambda \circ \omega_{H_i} )(\one/d)) \\
 \label{eq:monotonicity_proof3}&= D_2 ((\omega_{H_f} \circ \Lambda)(\hat\sigma) \| \one/d)\\
      &=R(\Lambda(\hat\sigma),H_f),
\end{align}
where \eqref{eq:monotonicity_proof1} follows since $\omega_H$ is a unital map, \eqref{eq:monotonicity_proof2} is a consequence of the data-processing inequality \eqref{eq:dpi}, and \eqref{eq:monotonicity_proof3} follows from \eqref{eq:monotonicity_proof0} and the fact that $\Lambda$ is a unital map.

We now turn to prove additivity of the resilience
\eqref{eq:additivity_resilience}. For the proof of Theorem~\ref{thm:mainresult}
one requires additivity of the resilience for given
states $\hat\sigma^{QR}=\hat\sigma^Q \otimes \hat\sigma^R$ that are stationary with respect to a non-interacting Hamiltonian $H_i^{QR}$.
In this case one finds
\begin{align}
  \omega_{H_i^{QR}}(\hat\sigma^{QR})&=\omega_{H_i^{QR}}(\hat\sigma^Q \otimes \hat\sigma^R)\\
  \label{eq:app:condition_Hamiltonians}&=\omega_{H^Q_i}(\hat\sigma^Q)\otimes \omega_{H^R_i}(\hat\sigma^R).
\end{align}
For the proof of Theorem~\ref{thm:mainresult2} one needs to show additivity for non-stationary initial states, under the extra assumption that
\begin{align}\label{eq:app:no_correlation_equilibration}
 \omega_{H^{QR}_i}=\omega_{H^Q_i} \otimes \omega_{H^R_i},
\end{align}
which in turns imply \eqref{eq:app:condition_Hamiltonians}. Hence, we conclude
that Eq.~\eqref{eq:app:condition_Hamiltonians} is fulfilled under the conditions of both Theorem~\ref{thm:mainresult} and \ref{thm:mainresult2}.
We can now use this equation together with additivity of the R\'{e}nyi divergences to show
additivity of the resilience as
\begin{align}
	&\R(\hat\sigma^Q \otimes \hat\sigma^R, H^{QR}_i)\nonumber\\
 &=D_2\big( \omega_{H^{QR}_i}(\hat\sigma^Q \otimes \hat\sigma^R )\| \one/d_{QR}\big)\\
 &=D_2\big( \omega_{H^Q_i}(\hat\sigma^Q)\otimes \omega_{H^R_i}( \hat\sigma^R) \|
  \one/d_{QR}\big)\\
  & =D_2\big( \omega_{H^Q_i}(\hat\sigma^Q) \| \one/{d_Q}\big)+D_2\big( \omega_{H^R_i}(\hat\sigma^R) \| \one/{d_R}\big) \\
 &=\R(\hat\sigma^Q,H^Q_i)+\R(\hat\sigma^R,H^R_i).
\end{align}

\subsection{\ref*{localquench}: Local quenches}

The statement of Corollary 2 is an immediate consequence of Theorem~\ref{thm:mainresult}. Consider a quantum many-body system
initially prepared in $\hat \sigma^Q$. Then one applies the local quantum channel $\Phi$ to a
\he{subsystem $X$ of $Q$ that is $d_X$-dimensional.} \he{In order to implement this quantum channel, one
considers an auxiliary system $R$ of dimension $d_R=d_X$ 
initially prepared in a mixed state $\hat\sigma^R$ so that 
\begin{equation}
	\Phi (\hat \sigma^Q) = \Tr_R\left(
	U (\hat \sigma^Q\otimes \sigma^R)U^\dagger\right)
\end{equation}
for a suitable unitary $U$ acting upon $R$ and $Q$. Such a state $\hat\sigma^R$ and unitary $U$ exist for any quantum channel $\Phi$ \cite{Nielsen2000}.
Since we can choose the Hamiltonian on $R$ freely, we can assume that $\hat\sigma^R$ is stationary.
The statement of the corollary then follows from Theorem~\ref{thm:mainresult} and the bound
\begin{equation}
	\R(\hat \sigma^R,H^R_i) \leq \log(d_R) = \log(d_X),
\end{equation}
which immediately follows from the definition of the effective dimension.
}

\subsection{\ref*{role_of_correlations}: The role of correlations}
In this section, we discuss in more detail the role of correlations and the proof of Theorem~\ref{thm:no_second_law}.
We begin with the proof of Theorem~\ref{thm:no_second_law}.
As stated in the main text, it relies on the following theorem, which we adapt from Ref.~\cite{Mueller2017}, whereby for any state $\hat\sigma$ we denote by $S(\hat\sigma)$ its von~Neumann entropy.

\begin{theorem}[Ref.~\cite{Mueller2017}]\label{thm:mueller}
For any two finite-dimensional states $\hat\sigma^Q$ and
$\hat\rho^Q$ of the same dimension with
$S(\hat\sigma^Q)\leq S(\hat\rho^Q)$, for any $\delta I>0$ and any $\epsilon>0$ there exists a
finite-dimensional state $\hat\sigma^R$ and a \he{mixture of unitaries $\Lambda$} such that
 \begin{enumerate}[(a)]
 \item \label{enum:correct_state_on_q} The channel produces the state $\hat\rho^Q$ on $Q$ to accuracy $\epsilon>0$,\begin{equation}
	 \norm{\hat\rho^Q - \tr_R(\Lambda(\hat\sigma^Q \otimes \hat\sigma^R))}_1 < \epsilon\je{.}
	 \end{equation}
 \item The state on $R$ after $\Lambda$ coincides with the state in which it was originally given:
  \begin{equation}\hat\sigma^R = \tr_Q(\Lambda(\hat\sigma^Q\otimes \hat\sigma^R)) \eqqcolon \hat\rho^R\je{.}	 \end{equation}
 \item The mutual information between $R$ and $Q$ after $\Lambda$ has acted is upper bounded by $\delta I$,
 \begin{equation}
  D_1(\Lambda(\hat\sigma^Q \otimes \hat\sigma^R) \| \hat\rho^Q\otimes\hat\rho^R) \leq \delta I.
  \end{equation}
 \end{enumerate}
 The Hilbert space dimension of $R$ may in general depend on both $\epsilon$ and $\delta I$.
\end{theorem}
\begin{proof}[Proof of Theorem~\ref{thm:no_second_law}]
For any fixed finite-dimensional Hilbert space, the above Theorem~\ref{thm:mueller} shows that
whenever a state $\rho^Q$ has higher von~Neumann entropy than
the state $\hat\sigma^Q$, then properties (\ref{enum:resource_unchaged}) and
(\ref{enum:final_correlations}) of Theorem~\ref{thm:no_second_law} are met. It
remains to show that it is always possible to find states $\hat\sigma^Q$ and
$\rho^Q$
that simultaneously fulfill $S(\hat\sigma^Q)\leq S(\rho^Q)$, have property~\mbox{(\ref{enum:change_of_resiliance})} of Theorem~\ref{thm:no_second_law}, and are such
that $\rho^Q$ does not equilibrate (property \mbox{(\ref{enum:non_equilibrating})}).

As required, we will hence now show that there
exist $\sigma_{(n)}^Q$ and $\rho_{(n)}^Q$ with the following properties:
\begin{enumerate}[(i)]
\item $S(\hat\sigma_{(n)}^Q) \leq S(\rho_{(n)}^Q)$,\label{cond:1}
\item $\Delta \R^{Q}_{(n)} \rightarrow \infty$ as $n\rightarrow \infty$,\label{cond:2}
\item $\rho_{(n)}^Q$ does not equilibrate.\label{cond:3}
\end{enumerate}
In the following let $d=D^n$ be the dimension of the Hilbert space of the $n$-partite system.
The initial states $\rho_{(n)}^Q$ of the non-equilibrating trajectories can be taken to be of the form
\begin{align}\label{eq:app:state}
 \rho_{(n)}^Q = a \proj{\Psi} + (1-a) \frac{\Pi}{d-2} ,
\end{align}
where $\ket\Psi = (\ket{E_1} + \ket {E_2})/\sqrt{2}$ is a superposition of two arbitrary energy-eigenstates of the Hamiltonian $H_{(n)}^Q$ (which can be chosen arbitrarily), $\Pi$ is the projector on the subspace orthogonal to $\ket{E_1}$ and $\ket{E_2}$, and $a \in [0,1]$.
This family of states has an entropy given by
\begin{align}
  S_1&\left(\rho_{(n)}^Q\right) \nonumber\\
     &= a S_1(\proj{\Psi}) + (1-a) S_1\left(\frac{\Pi}{d-2}\right) + H_2(a),\\
     &=(1-a)\,n\,\log(D) + (1-a)\,\log(1- 2 D^{-n}) + H_2(a),  \\
     &\approx (1-a)\,n\,\log(D) + H_2(a),
\end{align}
where $H_2(a) = - a\log(a) - (1-a)\log(1-a)$ is the binary entropy of $a$ and the error in the last approximation is exponentially small in $n$ for large $n$.
On the other hand, the effective dimension of $\rho_{(n)}^Q$ approaches a constant:
\begin{align}\label{eq:defffinal}
	\deff_{H_{(n)}^Q}\left(\rho_{(n)}^Q\right) =\frac{1}{ a^2 + \frac{(1-a)^2}{d-2}}\leq \frac{1}{a^2}.
\end{align}
In \ref{effective_dimensions}, this scaling is illustrated in more
detail.
It is clear that the states $\rho_{(n)}^Q$ do not equilibrate
since there will be Rabi-oscillations at frequency $E_1-E_2$ with amplitude $a$ independent of $n$
for all times.
This proves (\ref{cond:3}).

We now turn to constructing the family of given states $\hat\sigma_{(n)}^Q$.
Clearly, for any finite $n$ such states can be constructed. We thus now
focus on the case of arbitrarily large $n$. To fulfill condition
(\ref{cond:1}) for large enough $n$ it suffices to have
\begin{align}\label{eq:entropyscaling}
	S_1\left(\hat\sigma^{Q}_{(n)}\right) \leq (1-\tilde{a})\, n\, \log(D) ,
\end{align}
for any $\tilde{a}>a$.
To fulfill condition (\ref{cond:2}) it suffices to have
\begin{equation}\label{eq:divergingS2}
\lim_{n \rightarrow \infty} S_2(\hat\sigma_{(n)}^Q) =\infty.
\end{equation}
To see this, note that the resilience of the initial state can be written as
\begin{align}
	\R(\hat\sigma_{(n)}^Q,H_{(n)}^Q) = \log(d) - S_2(\hat\sigma_{(n)}^Q),
\end{align}
which holds for any Hamiltonian $H_{(n)}^Q$ for which $\sigma_{(n)}^Q$ is stationary.
Using this together with the bound on the effective dimension of the final
state~\eqref{eq:defffinal}, we obtain for the change of resilience
\begin{align}
	\Delta \R^Q &= \R\left(\rho_{(n)}^Q,H_{(n)}^Q\right) - \R\left(\hat\sigma_{(n)}^Q,H_{(n)}^Q\right) \\
              &\geq \log\left(\frac{d}{a^2}\right) - \R\left(\hat\sigma_{(n)}^Q,H_{(n)}^Q\right) \\
              &= S_2(\hat\sigma_{(n)}^Q) - 2\log(a).
\end{align}
Hence we see that the change in resilience becomes arbitrarily large as long as the
R\'{e}nyi-2 entropy of the given states on $Q$ diverges with $n$.
In general, the R\'{e}nyi-2 entropy $S_2$ is upper bounded by the von~Neumann entropy $S_1$, but can
be arbitrarily close to the latter.
Therefore a diverging $S_2$ is well compatible with Eq.~\eqref{eq:entropyscaling}.

We will now provide examples of states that indeed fulfill Eqs.~\eqref{eq:entropyscaling} and \eqref{eq:divergingS2} which concludes the proof.
There exist in fact many families of states fulfilling those conditions, but we present as an example a family of micro-canonical states $\hat\sigma_\Pi^{(n)}$, which is maximally mixed on a subspace $\Pi$ of dimension
$d^{\gamma}$ with $\gamma<1$.
This is the scaling that one expects for an
actual micro-canonical state in a local many-body system, since it leads to an
entropy scaling extensively with the system size:
\begin{align}
 S_1(\hat\sigma_\Pi^{(n)}) = \gamma\,n\,\log(D).
\end{align}
The constant $\gamma>0$ will of course depend on the effective temperature of the
state (i.e., the temperature of the canonical state with the same mean energy)
and can be made arbitrarily small.
Lastly, note that the R\'{e}nyi-2 entropy diverges for any value of $\gamma>0$
\begin{align}
S_2(\hat\sigma_\Pi^{(n)}) &=- \log \Tr ((\hat\sigma_\Pi^{(n)})^2)\\
&=- \log d^{-\gamma}= \gamma\, n\, \log (D).
\end{align}
This completes the proof.
\end{proof}

Another natural example of states fulfilling \eqref{eq:entropyscaling} and \eqref{eq:divergingS2} is given by stationary initial states which fulfill
exponential clustering of correlations and have non-maximal entropy density.
Indeed it has
been proven recently that states with exponential decay of correlations have
diverging effective dimension \cite{Farrelly2016} with respect to local Hamiltonians.

Finally, as a side note, let us show that the resilience in the state of the resource
$\R(\hat\sigma_{(n)}^R,H_{(n)}^R)$, and hence also its Hilbert space dimension, must diverge as $\epsilon\rightarrow 0$.
Suppose to the contrary that this would not be true.
Then there would exist constants $C$ such that
\begin{align}\label{eq:uniformboundresilience}
	\R(\hat\sigma_{(n)}^R,H_{(n)}^R) \leq C_n \quad \forall \epsilon.
\end{align}
For simplicity we now fix a system size and drop all the $n$-dependence.
We can use the resource $\hat\sigma^R$ to bring $m$ uncorrelated copies of $\hat\sigma^Q$ into a non-equilibrating state $\rho^{Q_1\cdots Q_m}$.
By combining Theorem~\ref{thm:mainresult} with the bound~\eqref{eq:uniformboundresilience} we would then obtain
\begin{align}
 C \geq \R (\rho^{Q_1\cdots Q_m},H^{Q_1\cdots Q_m}) - m \R(\hat\sigma^{Q_1},H^{Q_1}).
\end{align}
Since in each step, we only correlated one of the systems with $R$ by an amount at most $\epsilon$, the final mutual correlations between the different copies of $Q$ are also bounded by $\epsilon$: $I(Q_i:Q_j)\leq \epsilon$.
Furthermore, the above equation holds, by assumption, for all $\epsilon$.
We can then take the limit $\epsilon\rightarrow 0$.
In this limit we have
\begin{align}
 \lim_{\epsilon\rightarrow 0} \rho^{Q_1\cdots Q_m} = \rho^{Q_1} \otimes\cdots\otimes \rho^{Q_m}.
\end{align}
Since the resilience is super-additive on product states, as we show in Lemma~\ref{lemma:superadditivity} below, we would therefore obtain
\begin{align}
 C \geq m \left(\R(\hat\rho^{Q_1},H^{Q_1}) - \R(\hat\sigma^{Q_1},H^{Q_1})\right).
\end{align}
Since this equation holds for all $m$, we obtain a contradiction.

\begin{lemma}[Super-additivity on product states]\label{lemma:superadditivity}
The resilience is super-additive on product states of non-interacting systems:
\begin{align}
	\R &\left(\rho^Q\otimes \rho^R, H^{QR}\right) \geq \R(\rho^Q,H^{Q})+ \R(\rho^R,H^{R}), 
\end{align}
if $H^{QR} = H^Q\otimes \one + \one\otimes H^R$.
\begin{proof}
By using the integral-representation of $\omega_{H^Q}$ as a time-average, it follows immediately that local and global dephasing commutes. We thus have
\begin{align}
	\omega_{H^Q}\otimes \omega_{H^R} &= \omega_{H^{QR}} \circ\left(\omega_{H^Q}\otimes \omega_{H^R}\right) \\
					 &= \left(\omega_{H^Q}\otimes \omega_{H^R}\right)\circ \omega_{H^{QR}}.
\end{align}
Using the data-processing inequality, we then obtain
\begin{align}
  \R&(\rho^Q\otimes \rho^R, H^{QR})\nonumber\\
&= D_2\left(\omega_{H^{QR}}(\rho^Q\otimes \rho^R) \| \one/d_{QR}\right)\\
&\geq D_2\left((\omega_{H^Q}\otimes \omega_{H^R})\circ \omega_{H^{QR}}(\rho^Q\otimes\rho^R) \| \one/d_{QR}\right)\\
&= D_2\left(\omega_{H^Q}\otimes \omega_{H^R}(\rho^Q\otimes\rho^R) \| \one/d_{QR}\right)\\
& = \mc R(\rho^Q,H^Q) + \mc R(\rho^R,H^R)
\end{align}
which finishes the proof. 
\end{proof}
\end{lemma}

\subsection{\ref*{effective_dimensions}: Comparison of effective dimensions}
As explained in the main text, our results are based on the effective dimension
as a suitable quantifier of the equilibrating properties of a system, which is
justified by bounds \eqref{eq:equilibrationonaverageforexpectationvalues} and
\eqref{eq:equilibration}. Nonetheless, there exists in the literature on the
topic similar bounds formulated in terms of quantities other than effective
dimension. That said, it seems natural to investigate whether alternative
definitions of resilience can be put forward based on those other bounds,
or even based on bounds which, although not yet shown to hold, are conceivably
true.

First let us consider bounds of the form \eqref{eq:equilibrationonaverageforexpectationvalues} and \eqref{eq:equilibration} based on the second largest of the eigenvalues of the time averaged, dephased state \cite{Reimann12}. This bound takes the form
\begin{align} 
  \he{\Var_H(A,\rho) \leq} \|A\|^2\ 3 \: \psecond ,\nonumber
\end{align}
where $\psecond$ is the second largest eigenvalue of $\omega_H(\rho)$. Indeed, in some situations this bound is stronger than the one based on the effective dimension. However, our bounds on the cost of preparing states that do not equilibrate cannot be analogously carried over using $\psecond$, instead of the effective dimension, since $\psecond$ is not a monotone under mixtures of unitaries \cite{Nielsen2000}. That is, it is not possible to construct an alternative measure of resilience based on $\psecond$ which fulfills monotonicity and additivity as given by \eqref{eq:monotonicity} and \eqref{eq:additivity_resilience}. 

We now discuss, in the light of our results, the possibility of strengthening bounds on equilibration. Note that combining Eqs. \eqref{eq:equilibrationonaverageforexpectationvalues}, \eqref{eq:equilibration} and \eqref{eq:connectiontsalliswithdeff} we see that 
\begin{align} \label{eq:bound_equilibration_alpha2}
\Var_H(A,\rho) \leq \|A\|^2\  e^{ -S_{2} (\omega_H (\rho))},
\end{align}
where $S_2$ is the R\'enyi-2 entropy. Here, the R\'enyi-$\alpha$ entropy is defined as
\begin{align}
	S_\alpha(\rho) \coloneqq \log(d) - D_\alpha(\rho \| \one/d) = \frac{1}{1-\alpha}\log(\tr(\rho^\alpha)). \nonumber
\end{align}
Eq.~\eqref{eq:bound_equilibration_alpha2} strongly suggests to consider the possibility of having bounds in terms of R\'enyi entropies $S_{\alpha}$ with $\alpha\neq 2$. That is, bounds of the form
\begin{align}\label{eq:general_bound_alpha}
\Var_H(A,\rho) \leq \|A\|^2\  g\left( S_{\alpha} (\omega_H (\rho))\right),
\end{align}
where $g$ is any function so that $\lim_{x\rightarrow \infty} g(x)=0$. Let us first consider the case $\alpha=1$, which yields the usual von~Neumann entropy $S_1$. We show now that a bound of the form \eqref{eq:general_bound_alpha} for $\alpha=1$ is impossible, since the state constructed in Eq.~\eqref{eq:app:state} provides a counter-example.
The energy distribution of this non-equilibrating state (see the discussion below Eq.~\eqref{eq:defffinal}) is 
\begin{align} \label{eq:twinpeakdistribtion}
	p_1 &= p_2 = \frac{a}{2}, & \forall\, 2 < k \leq d \colon p_k = \frac{1-a}{d-2} .
\end{align}
Using the entropy scaling given in Eq.\ \eqref{eq:entropyscaling} it
follows that $S_1$ grows linearly with the system's size, although the system does not equilibrate. The implications
of the impossibility of a bound of the type
\eqref{eq:general_bound_alpha} with $\alpha=1$ in relation with Theorem \ref{thm:no_second_law} are laid out in the next section.

The previous counter example together with the fact that $S_{\alpha}(\rho) \geq S_{\alpha'}(\rho)$ if $\alpha \leq \alpha'$ (see Ref.\  \cite{Tomamichel2016}) implies that it is also impossible to find bounds of the form \eqref{eq:general_bound_alpha} for $\alpha\leq 1$. Altogether this suggests a possibility of improving current bounds of equilibration \eqref{eq:bound_equilibration_alpha2} by formulating it in terms of $S_{\alpha}$ with $1<\alpha <2$. We hope this triggers further research on tighter equilibration bounds based on these insights.

\subsection{\ref*{shannon}: Theorem \ref{thm:no_second_law} and its relation to super-additivity}
As mentioned above, our results also imply that meaningful
equilibration bounds cannot be derived in terms of the von~Neumann entropy of the
time-averaged state.
In this section, we explain in detail how this relates with Theorem \ref{thm:no_second_law}. To do so, let us assume the contrary, that is, that a bound of the form \eqref{eq:general_bound_alpha} with $\alpha=1$ would exist.
In this case, a sensible definition of resilience would be given by
\begin{align}\label{eq:app:alternative_resilience}
	\tilde{\R} (\rho,H) \coloneqq D_1(\omega_H(\rho)\| \one/d).
\end{align}
As can be easily seen from the properties of the quantum relative
entropy, $\tilde{\R}$ is also additive and fulfills the data-processing
inequality.
Hence, Theorem~\ref{thm:mainresult} also holds in terms of
$\tilde{\R}$.
However, $\tilde{\R}$ has the additional property of being
\emph{super-additive}, meaning that if we consider a bipartite, non-interacting
system in a, possibly correlated, state $\rho^{QR}$, we have
\begin{align}
	\tilde{\R}(\rho^{QR},H^Q\otimes\one+H^R\otimes\one) \geq \tilde{\R}(\rho^Q,H^Q) + \tilde{\R}(\rho^R,H^R).
\end{align}
This follows from the super-additivity of the quantum relative entropy (or
equivalently from the sub-additivity of the von~Neumann entropy) \cite{Nielsen2000}.
We can use this property to obtain an even stronger result than Theorem~\ref{thm:mainresult} by essentially the same calculation:
\begin{align}
\tilde{\R}(\sigma^Q,H^Q) + \tilde{\R}(\sigma^R,H^R) &\geq \tilde{\R}(\rho^{QR},H^R+H^Q) \\
	 &\geq \tilde{\R}(\rho^{Q},H^Q) + \tilde{\R}(\rho^R,H^R).
\end{align}
Rearranging, we obtain the second-law like inequality
\begin{align}
	-\Delta\tilde{\R}^R \geq \Delta\tilde{\R}^Q.
\end{align}
Hence, one concludes that the results of Theorem \ref{thm:no_second_law}, showing a violation of this second-law like inequality, imply that it is not possible to find equilibration bounds of the form of Eqs. \eqref{eq:equilibrationonaverageforexpectationvalues} and \eqref{eq:equilibration} in terms of a sub-additive entropy, or equivalently, a super-additive divergence.

\clearpage
\end{document}